\documentclass[journal,twoside,web]{ieeecolor}
\usepackage{lcsys}
\usepackage{textcomp}
\usepackage{amsthm}
\usepackage{amsmath,amssymb,amsfonts}
\allowdisplaybreaks
\usepackage{graphicx}  
\usepackage{cite}
\usepackage{nicefrac}
\usepackage[english]{babel}
\usepackage{cases}
\usepackage{algorithm, algorithmic}
\usepackage{hyperref}
\hypersetup{
    colorlinks = true,
    citecolor = {blue},
    linkcolor=blue
}

\newtheorem{theorem}{Theorem}
\newtheorem{lemma}{Lemma}

\newtheorem*{assumption*}{\assumptionnumber}
\providecommand{\assumptionnumber}{}
\makeatletter
\newenvironment{assumption}[1]
{%
	\renewcommand{\assumptionnumber}{A#1}%
	\begin{assumption*}%
		\protected@edef\@currentlabel{A#1}%
	}
	{%
	\end{assumption*}
}
\makeatother
\def \newtext{}

\newcommand{\tr }{^\mathrm{T}}

\newcommand{\abs}[1]{\left\lvert#1\right\rvert}
\newcommand{\norm}[1]{\big\lVert#1\big\rVert}
\newcommand{\indnorm}[1]{{\left\vert\kern-0.25ex\left\vert\kern-0.25ex\left\vert #1 
		\right\vert\kern-0.25ex\right\vert\kern-0.25ex\right\vert}}
\newcommand{\defeq}{\doteq}
\newcommand{\confreg}[1][2]{\mathcal{C}_{#1}}
\newcommand{\outappr}[1][2]{\mathcal{O}_{#1}}

\newcommand{\BR}{\mathbb{R}}
\newcommand{\BE}{\mathbb{E}}
\newcommand{\BP}{\mathbb{P}}

\newcommand\independent{\protect\mathpalette{\protect\independenT}{\perp}}
\def\independenT#1#2{\mathrel{\rlap{$#1#2$}\mkern2mu{#1#2}}}

\def\hyph{-\penalty0\hskip0pt\relax}

\title{ 
Finite-Sample 
Identification {\newtext of
Linear Regression Models}
with Residual-Permuted Sums
}

\author{Szabolcs Szentp{\'e}teri$^{1}$, and Bal{\'a}zs Csan{\'a}d Cs{\'a}ji$^{1,2}$, \IEEEmembership{Member, IEEE} %
\thanks{*This research was supported by the European Union within the framework of the National Laboratory for Autonomous Systems, RRF-2.3.1-21-2022-00002; and by the TKP2021-NKTA-01 grant of the National Research, Development and Innovation Office (NRDIO), Hungary.}%
\thanks{$^{1}$Sz.~Szentp{\'e}teri and B.~Cs.~Cs\'aji are with Institute for Computer Science and Control (SZTAKI), Hungarian Research Network (HUN-REN), Budapest, Hungary,\, 
{\tt\small szabolcs.szentpeteri@sztaki.hu},\, {\tt\small balazs.csaji@sztaki.hu}}%
\thanks{$^{2}$B.~Cs.~Cs\'aji is also with Department of Probability Theory and Statistics, E\"otv\"os Lor\'and University (ELTE), Budapest, Hungary}%
}

\begin{document}

\maketitle
\thispagestyle{empty}
\pagestyle{empty}

\begin{abstract}

This letter studies a {distribution-free}, {finite-sample} data perturbation (DP) method, the {Residual-Permuted Sums} (RPS), which is an alternative of the Sign-Perturbed Sums (SPS) algorithm, to construct {confidence regions}. While SPS assumes independent (but potentially time-varying) noise terms which are symmetric about zero, RPS gets rid of the symmetricity assumption, but assumes i.i.d.\ noises. The main idea is that RPS permutes the residuals instead of perturbing their signs. 
This letter introduces RPS in a flexible way, which allows various design-choices. 
RPS has {exact} finite sample {coverage probabilities} and we provide the first proof that these permutation-based confidence 
regions 
are uniformly strongly consistent under general assumptions. This 
means that the RPS regions almost surely shrink around the true 
parameters as the sample size increases. The {ellipsoidal outer-approximation} (EOA) of SPS is also extended to RPS, and the effectiveness of %
RPS 
is validated by numerical experiments, as well.

\end{abstract}


\section{INTRODUCTION}

\IEEEPARstart{E}{stimating}
dynamical systems based on empirical data is a fundamental problem in system identification, 
machine learning and statistics. Classical results in the aforementioned areas, such as prediction error methods, typically provide point estimates with asymptotically guaranteed confidence regions \cite{Ljung1999}. However, in practical problems, where the robustness of the solution is a crucial aspect, confidence sets with finite sample guarantees are highly desirable. Due to these reasons, in the recent years, significant emphasis was given to the non-asymptotic theory of system identification \cite{Ziemann2023}. 

{\newtext 
A possible approach to build confidence regions with finite sample guarantees for i.i.d.\ samples is to utilize randomized 
hypothesis tests, e.g., Monte Carlo or bootstrap methods \cite{lehmann2022testing}. 
For linear regression problems, a random permutation based test was presented in \cite{Nyblom2015}, however, it builds on asymptotic approximations, hence it lacks finite sample guarantees.}

{\newtext An important identification method that uses the ideas behind randomized tests and can build exact confidence regions for finite sample sizes with distribution-free guarantees is Sign-Perturbed Sums (SPS) \cite{Csaji2015}. Later SPS was generalized, and Data Perturbation (DP) methods were introduced in \cite{kolumban2015perturbed}.}

The core idea of SPS is to perturb the signs of the residuals in the normal equations, assuming that the measurement noises are independent and symmetric about zero. Then, based on the rank statistics computed from these perturbed quantities, SPS 
builds a confidence region around the least-squares (LS) estimate \cite{Csaji2015}.
DP methods generalize this idea in a way that they allows different types of perturbations depending on the characteristics of the observation noises. As a prime example, a permutation-based DP method was introduced in \cite{kolumban2015perturbed}, where the residuals are permuted instead of sign-perturbed. 
This requires exchangeable 
noises, but symmetricity is not needed.

Another approach to relax the symmetricity assumption of SPS is LAD-SPS which builds confidence regions around the least-absolute-deviation (LAD) estimate 
under the assumption that the conditional medians of the noises are zero \cite{Care2016}. 
In all of the above mentioned algorithms, the coverage probability of the true parameters can be exactly guaranteed for any finite sample size, and the confidence set is given by an indicator function that can be queried at any parameter vector. 

{\newtext For 
linear regression problems, the (uniform) strong consistency of SPS was proven in \cite{Weyer2017}, which means that the SPS regions almost surely shrink around the true system parameters as the sample size increases. 
In addition, a compact representation of SPS confidence sets, given by an ellipsoidal outer-approximation (EOA) algorithm, was proposed in \cite{Csaji2015}.}

As so far consistency was only proven for SPS, it remained an open question whether other types of DP methods, such as permutation-based constructions, are  consistent, as well.
In this letter, we propose the Residual-Permuted Sums (RPS) method for linear regression problems, which is a generalization of the original permutation-based approach of  \cite{kolumban2015perturbed}. 
We rigorously prove the (uniform) strong consistency of RPS 
under general statistical assumptions, and also
extend the EOA of SPS to RPS.
Finally, simulation experiments are presented that demonstrate the effectiveness of RPS  by comparing it with SPS and asymptotic confidence ellipsoids.
\section{Problem Setting}\label{sec:problem_setting}
This section specifies the addressed linear regression problem and introduces our main assumptions on the system.
\subsection{Data Generation}
Consider the following linear regression problem
\begin{equation}\label{equ:system}
    Y_t \,\defeq\, \varphi_t\tr \theta^* + W_t,
\end{equation}
{\newtext for $t \in [n]\defeq\{1, \dots, n\}$,} where $\varphi_t$ is a $d$-dimensional random regressor, $Y_t$ is the scalar output, $W_t$ is the (random) scalar noise and $\theta^*$ is the $d$-dimensional (constant) true parameter.
We are given a sample of size $n$ which consists of regressor vectors (inputs) $\varphi_1, \dots, \varphi_n$ and outputs $Y_1, \dots, Y_n$.

\subsection{Assumptions}
Our assumptions on the noises and the regressors are
\begin{assumption}{1}\label{assu:noise}
The noise terms $\{W_t\}$ are independent and identically distributed (i.i.d.) with finite fourth moments: $\mathbb{E}\big[W_0^4\big] < \infty$.
\end{assumption}
Note that these assumptions on the noises are very mild, as most strong consistency results 
assume independence and require moment conditions from the noise terms.
Also note that unlike SPS \cite{Csaji2015}, the noises do not have to be symmetric about zero, nor do they need zero mean; however their i.i.d.\ nature is essential to ensure exact coverage probabilities.

\begin{assumption}{2}\label{assu:regressor}
{\newtext Regressors $\{\varphi_t\}$ have uniformly bounded fourth moments, $\forall\, t :
\mathbb{E}\big[\|\varphi_{t}\|^4\big] \leq \kappa < \infty$, and they
are $\ell$\hyph independent: $\forall\, t : \varphi_t \independent \{\varphi_k\}_{\abs{t-k}\geq \ell}$, where ``$\independent$'' denotes independence.}
\end{assumption}
{\newtext A consequence of $\ell$-independent regressors is that our analysis covers 
FIR and Generalised FIR models\cite{Ljung1999}. Also note that the independence of $\{\varphi_t\}$ and $\{W_t\}$ is not assumed.}

\subsection{Co-regressor construction}
We introduce a co-regressor based construction that is used by the proposed RPS algorithm. The motivation for using co-regressors is
to cover various design-choices.
We denote the $d$-dimensional random co-regressor by $\psi_t$, and assume that
\begin{assumption}{3}\label{assu:co-regressor_noise}
The co-regressor vectors $\{\psi_t\}$ are {\newtext $\ell$-independent} with uniformly bounded fourth moments, i.e.\ for every $t$,  we have {\newtext $\mathbb{E}\big[\|\psi_t\|^4\big] \leq \kappa < \infty$}, furthermore $\{\psi_t\}$ is independent of the noise sequence $\{W_t\}$, and for every $t:\mathbb{E}\big[\psi_tW_t\big] = 0$.
\end{assumption}
\begin{assumption}{4}\label{assu:co-regressor_regressor}
{\newtext We have}
{\newtext \begin{align}
    \forall\, i,j: \abs{\hspace{0.3mm}i-j}\geq \ell: \psi_i \independent \varphi_j
    \quad \text{and}\quad \mathbb{E}\big[\psi_i\varphi_j\tr\big] = 0,
\end{align}
\begin{align}
    \forall\, k,l: \forall\, i,j : \abs{\hspace{0.3mm}i-j}< \ell: \mathbb{E}\big[\psi_{i,k}^4\varphi_{j,l}^4\big] \leq \kappa < \infty,
\end{align}}
furthermore, the condition below holds almost surely\vspace{-0.5mm}
\vspace{-1mm}
\begin{align}\label{equ:def_Vn}
    \lim_{n \to \infty}\frac{1}{n}\sum_{t=1}^n\psi_t\varphi_t\tr \defeq \lim_{n \to \infty}V_n = V
    \succ 0,
\end{align}
{\newtext where ``$\succ 0$'' denotes positive definiteness.}
\end{assumption}

\begin{assumption}{5}\label{assu:R}
{\newtext There are user-chosen (random) matrices $\{R_n\}$, such that $\{R_n\} \independent \{W_t\}$, $\{R_n\}$ are
positive semidefinite, and there is a 
positive definite $R$,
such that almost surely}
\begin{align}\label{equ:def_Rn}
\lim_{n \to \infty}R_n = R \succ 0,
\end{align}
\end{assumption}
\smallskip
Matrices $\{R_n\}$ are only introduced to allow the reshaping of RPS regions, however, $R_n = I_n$ is also a valid choice.

There are several possible design choices for co-regressors that can satisfy \ref{assu:co-regressor_noise} and \ref{assu:co-regressor_regressor}, here we list some of them. In case $\BE\left[\varphi_t\right]$ is known, {\newtext and the regressors are independent of the noises,}
one can simply use $\psi_t \defeq \varphi_t-\BE[\varphi_t]$ as a co-regressor. In case $\BE\left[\varphi_t\right]$ is unknown, one can replace it with an estimate, $\psi_t \defeq \varphi_t-\zeta_t$, where $\BE\left[\zeta_t\right]=\BE\left[\varphi_t\right]$,
for example, an independent copy of $\varphi_t$ can be used as $\zeta_t$. Finally, one can also apply $\psi_t \defeq f(\varphi_t)$, where $f$ is a %
suitable function.

Let us illustrate this latter option with signed-regressors often used in adaptive filtering
\cite{sayed2003fundamentals}. Let $d=1$ and $Y_t \defeq \varphi_t\, \theta^* + |\varphi_t| N_t,$ where $\varphi_t \independent N_t$ and $\varphi_t, N_t \sim \mathcal{N}(0,1)$.
Then, $W_t \defeq |\varphi_t| N_t$ is not independent of $\varphi_t$. However, as $\text{sign}(\varphi_t)$ and $|\varphi_t|$ are independent, we can use $\psi_t \defeq \text{sign}(\varphi_t)$ as a co-regressor, since this ensures the independence of $\psi_t$ and $W_t$.
\section{The Residual-Permuted Sums Algorithm}\label{sec:SPS_overview}
In this section, we introduce the Residual-Permuted Sums algorithm. The method is a generalization of the permutation-based hypothesis test proposed in \cite{kolumban2015perturbed}. It consists of two parts, in the first part the main parameters and the random permutations are computed, while the second part decides whether a given parameter $\theta$ is included in the confidence region. The first part is given by Algorithm  \ref{alg:rps_init} and the second is presented by Algorithm \ref{alg:rps_indicator}. Using this construction, the $p$-level RPS confidence region can be defined as follows
\begin{equation}
        \confreg[p,n]\, \defeq\, \{\hspace{0.3mm}\theta \in \mathbb{R}^d\text{ : RPS-Indicator}(\theta) = 1\hspace{0.3mm}\}.
        \vspace{-2mm}
\end{equation}
\begin{algorithm}[t]
    \caption{Pseudocode: RPS-Initialization\,$(\hspace{0.2mm}p\hspace{0.2mm})$}
    \label{alg:rps_init}
	\begin{algorithmic}[1]
 		\STATE Given a (rational) confidence probability $p \in (0,1)$, set integers $m > q >0$ such that $p = 1 - q/m$.
  		\STATE Choose a positive semidefinite matrix $R_n$ and find its principal square root $R_n^{1/2}$, i.e., the p.s.d. matrix with
        \begin{align*}
            R_n^{1/2}R_n^{1/2\mathrm{T}} = R_n.
        \end{align*}
        \vspace{-5mm}
		{\newtext \STATE For all
        $i \in [n]$,
        generate (independent) uniform random permutations $\sigma_{i,n}$ of 
        $[n]$, that is, each of the $n!$ possible permutations has probability $1/(n!)$ to be selected.
		\STATE Generate a random permutation $\pi$ of 
        $\{0,\dots, m - 1\}$  uniformly, i.e., each permutation has probability $1/(m!)$}
	\end{algorithmic}
\end{algorithm}
\begin{algorithm}[t]
	\caption{Pseudocode: RPS-Indicator\,$(\hspace{0.2mm}\theta\hspace{0.2mm})$}
    \label{alg:rps_indicator}
	\begin{algorithmic}[1]
		\STATE Compute the prediction errors for $\theta:$ for $t \in [\hspace{0.3mm}n\hspace{0.2mm}]$ let
        \vspace{-0.5mm}
        \begin{align*}
            \varepsilon_t(\theta)\, \defeq\, Y_t - \varphi_t\tr \theta.
        \end{align*}
        \vspace{-4.5mm}
		\STATE Evaluate for $i \in [\hspace{0.3mm}m-1\hspace{0.2mm}]$ the following functions:
        \vspace{0mm}
        \begin{align*}
             &S_0(\theta)\, \defeq\, R_n^{-\frac{1}{2}}\frac{1}{n}\sum_{t=1}^{n}\psi_t\varepsilon_t(\theta),\\
             &S_i(\theta)\, \defeq\, R_n^{-\frac{1}{2}}\frac{1}{n}\sum_{t=1}^{n}\psi_t\varepsilon_{\sigma_{i,n}(t)}(\theta).
        \end{align*}
        \vspace{-1mm}
		\STATE Compute the rank 
        of $\|S_0(\theta)\|^2$ among $\{\|S_i(\theta)\|^2\}:$
        \vspace{-1mm}
            \begin{equation*}
                \mathcal{R}(\theta)\, \defeq\, \left[\,1+\sum_{i=1}^{m-1}\mathbb{I}\left(\norm{S_0(\theta)}^2 \succ_{\pi} \norm{S_i(\theta)}^2\right)\,\right]\!,
                \vspace{-0.5mm}
            \end{equation*}
        where ``$\succ_{\pi}$'' is ``$>$'' with random tie-breaking, i.e., 
        $\|S_k(\theta)\|^2 \succ_{\pi} \|S_j(\theta)\|^2$ if and only if $(\|S_k(\theta)\|^2 > \|S_j(\theta)\|^2) \lor
        (\|S_k(\theta)\|^2 = \|S_j(\theta)\|^2 \,\land\, \pi(k) > \pi(j))$.           
 		\STATE Return 1 if\, $\mathcal{R}(\theta) \leq m - q$, otherwise return 0.
	\end{algorithmic}
\end{algorithm}
\vspace{-3mm}
\subsection{Exact coverage of RPS confidence regions}
The exact coverage probability of the permutation-based variant of SPS was proved in \cite{kolumban2015perturbed}, for the case of deterministic regressors.
This result can be extended to cover the exact confidence of RPS regions under our assumptions: 
\begin{theorem}\label{thm:exact_confidence}
    {\newtext Assuming $\{W_t\}$ are i.i.d., and $R_n$, $\{\psi_t\}$ are independent of $\{W_t\}$, the coverage probability of the constructed confidence region $\confreg[p,n] $ is exactly $p$, that is}
    \begin{equation}
            \BP(\theta^* \in \confreg[p,n]) \,=\, 1-\frac{q}{m} \,=\, p.
            \vspace{0.5mm}
    \end{equation}
\end{theorem}
\begin{proof}
{\newtext The exact coverage of
RPS regions
can be proven very similarly to the proofs in \cite{Csaji2015, kolumban2015perturbed,volpe2015sign}: by showing that $\{ \| S_i(\theta^*)\|^2\}_{i=0}^{m-1}$ are exchangeable.
First, for deterministic $R_n$ and $\{\psi_t\}$, it can be shown that $\{ \| S_i(\theta^*)\|^2\}_{i=0}^{m-1}$ are conditionally i.i.d.\ (thus also exchangeable) w.r.t.\ the ordered noises, i.e., the $\sigma$-algebra generated by $(W_{(1)}, \dots, W_{(n)})$. This can be generalized for random $R_n$ and $\{\psi_t\}$ by using that they are independent of $\{W_t\}$ and the law of total expectation, i.e.\ by conditioning on 
$R_n$ and $\{\psi_t\}$, as well.}
\end{proof}

\section{Strong Consistency of RPS Regions}\label{sec:strong_consistency}
In this section, we present one of the main contributions of this letter, the proof that the confidence regions generated by RPS are strongly consistent. First, we prove a lemma that plays a key part in the proof of the main theorem.
\subsection{Permutation lemma}
The next lemma is a strong law of large numbers (SLLN) for randomly permuted sequences. The main idea behind its proof is to extend Cantelli's SLLN \cite{shiryaev2021probability}.
Note that SLLN type theorems for permuted sequences in the literature mainly focus on a single exchangeable sequence, however, we have a new permutation for every $n$, i.e., 
a double-indexed sequence.
\begin{lemma}\label{lemma:permuted_sum}
    {\newtext Let $\{X_i\}$ and $\{Y_i\}$ be sequences of $\ell$-independent random variables with $\BE[X_i^aY_j^b] \leq \kappa_0 < \infty$ for all $\abs{j-i}<\ell$ and $a, b \in \mathbb{N}_0$ satisfying $0 \leq a, b \leq 4$ and $(a=b$ or $a+b \leq 4)$.
    Furthermore, for $\abs{j-i}\geq\ell$ let $X_i$ and $Y_j$ be independent and $\BE[X_iY_j] =0$. Let $\{\sigma_n\}$ be 
    independent,
    where $\sigma_n$ is a uniform random permutation of $[n]$.
    Then, we have}
    \vspace{-1mm}
\begin{align}\label{equ:lemma1_claim}
    \frac{1}{n}\sum_{i=1}^{n}X_iY_{\sigma_n(i)} \xrightarrow{\text{a.s.}} 0\qquad(\text{as}\;\;n\to \infty).
\end{align}
\end{lemma}
\begin{proof}
{\newtext 
For every $s \in \mathbb{N}$, let $J_s \doteq \{s,\, s+\ell,\, s + 2\,\ell, \dots\}$. Then,
for each $s$, $\{X_{j}\}_{j\in J_s}$ is an independent sequence. Let $I_s \doteq I_s^n \doteq J_s \cap [n]$, hence if $s\leq \ell$, $\lfloor n/\ell\rfloor \leq \abs{I_s} \leq \lceil n/\ell\rceil$.

By the (first) Borell-Cantelli lemma \cite{shiryaev2016probability}, \eqref{equ:lemma1_claim} holds if
\vspace{-0.5mm}
\begin{align}
    \label{eq:BClemma}
    &\sum_{n=1}^\infty\mathbb{P}\Biggl\{\bigg\vert\frac{1}{n}\sum_{i=1}^{n}X_iY_{\sigma_n(i)}\bigg\vert \geq \varepsilon\Biggr\}\notag\\
    & \leq\sum_{n=1}^\infty\mathbb{P}\Biggl\{\bigg\vert\frac{1}{n}\sum_{i\in I_1}X_iY_{\sigma_n(i)}\bigg\vert + \cdots + \bigg\vert\frac{1}{n}\!\sum_{i\in I_{\ell}}X_iY_{\sigma_n(i)}\bigg\vert\geq \varepsilon\Biggr\}\notag\\
    &\leq \sum_{n=1}^\infty\sum_{s=1}^{\ell}\mathbb{P}\Biggl\{\bigg\vert\frac{1}{n}\sum_{i\in I_{s}}X_iY_{\sigma_n(i)}\bigg\vert \geq \frac{\varepsilon}{\ell}\Biggr\} < \infty,
\end{align}
for any $\varepsilon > 0$, 
where %
we applied the triangular inequality and the union bound.}
By using the (generalized) Chebyshev inequality, the convergence of the series in \eqref{eq:BClemma} follows from
\vspace*{-1mm}
{\newtext \begin{align}
    \label{eq:Chebyshev}
    \frac{\ell^4}{\varepsilon^4}\sum_{n=1}^\infty\sum_{s=1}^{\ell}\mathbb{E}\hspace{0.5mm}\Big\vert\frac{1}{n}\sum_{i\in I_{s}}X_iY_{\sigma_n(i)}\Big\vert^4 < \infty.
\end{align}}
Therefore, our goal will be to show \eqref{eq:Chebyshev}. Let us expand\vspace{-1mm}
\begin{align}
    \label{equ:expanded_sum}
    &\Big|\sum_{i\in I_s}X_iY_{\sigma_n(i)}\Big|^{4} =\,\sum_{i\in I_s}X_i^4Y_{\sigma_n(i)}^4\notag\\
    &+6\,\sum_{i,j\in I_s,i<j}X_i^2Y_{\sigma_n(i)}^2X_j^2Y_{\sigma_n(j)}^2\notag\\
    &+12\!\sum_{i,j,k\in I_s,i<j<k}X_i^2Y_{\sigma_n(i)}^2X_jY_{\sigma_n(j)}X_kY_{\sigma_n(k)}\notag\\
    &+24\!\!\!\sum_{i,j,k,l\in I_s,i < j <k<l}\!\!X_i Y_{\sigma_n(i)}X_j  Y_{\sigma_n(j)}X_k Y_{\sigma_n(k)}X_l Y_{\sigma_n(l)}\notag\\
    &+8\,\sum_{i,j\in I_s,i\neq j}X_i^3Y_{\sigma_n(i)}^3X_jY_{\sigma_n(j)}.
\end{align}
Now, we will take the expectation of \eqref{equ:expanded_sum}, term by term.
{\newtext In the first two terms,
for any permutation, none of the summed expected values are zero, and each one of them can be upper bounded by a corresponding power of $\kappa_0$ using our independence and moment assumptions. Note that $a=b=0$ is also allowed which ensures that $\kappa_0 \geq 1$. Then,
\vspace{-1mm}
\begin{align}
    &\sum_{s=1}^{\ell}\sum_{i\in I_s}\mathbb{E}\left[X_i^4Y_{\sigma_n(i)}^4\right] \leq \ell\lceil n/\ell\rceil\kappa_0^2  = \mathcal{O}(n),\\
    &\sum_{s=1}^{\ell}\sum_{i,j\in I_s,i<j}\!\!\!\!\mathbb{E}\left[X_i^2Y_{\sigma_n(i)}^2X_j^2Y_{\sigma_n(j)}^2\right] \leq \ell\binom{\lceil \frac{n}{\ell}\rceil}{2}\kappa_0^4 = \mathcal{O}(n^2).\notag
\end{align}

In order to upper bound the expectation of the third term, we introduce the $\ell$-neighbourhood of index $i$ as $N(i) \defeq \{j: \abs{\hspace{0.3mm}i-j} <\ell \}$. 
If $j \notin N(i)$, then $\BE[X_iY_j] =0$ and $X_i\independent Y_j$, and consequently,
the summed expectations in the third term can be nonzero only: $a$) if the $\ell$-neighbourhoods of $j$ and $k$ each contains at least one of $\sigma_n(i),\sigma_n(j),\sigma_n(k)$ or $b$) $\sigma_n(j)$ and $\sigma_n(k)$ are in the $\ell$-neighbourhood of $i$ or $\sigma_n(i)$, or if they in the $\ell$-neighbourhood of each other. More precisely, by using the law of total probability and introducing the events $A \defeq A_i(j,k) \defeq \{\hspace{0.3mm}\sigma_n(i) \in N(j) \cap N(k)\hspace{0.3mm}\}$ and $B_j(k) \defeq B_j(i,k) \defeq \{\hspace{0.3mm}\sigma_n(j) \in N(i) \cup N(\sigma_n(i)) \cup N(\sigma_n(k))\hspace{0.3mm}\}$, 
\vspace{-1mm}
\begin{align}
    &\sum_{s=1}^{\ell}\sum_{\substack{i,j,k\in I_s\\i<j<k}}\!\!\mathbb{E}\big[X_i^2Y_{\sigma_n(i)}^2X_jY_{\sigma_n(j)}X_kY_{\sigma_n(k)}\,\big\lvert A \lor(B_j(k)\land \notag\\
    &B_k(j))\big]\mathbb{P}(A \lor (B_j(k)\land B_k(j)))\leq\ell\binom{\lceil \frac{n}{\ell}\rceil}{3}\kappa_0^4\,\cdot\\
    &\cdot\Bigg[\Bigg(\frac{\ell\,(n-1)\,(n-2)}{n(n-1)(n-2)}\Bigg)+\Bigg(\frac{4\ell\,6\ell\,(n-2)}{n(n-1)(n-2)}\Bigg)\Bigg] = \mathcal{O}(n^2),\notag
\end{align}
where the upper bound $\kappa_0^4$ on the conditional expectation follows from our moment assumptions and the repeated application of the 
Cauchy–Schwarz inequality. The argument behind the definition of event $A$ is to cover the possibilities given in the description of point $a$). Note that it can be that $N(j) \cap N(k)\neq \emptyset$ due to the construction of $I_s$, and our definition of $A$ gives the least constraints on %
$\sigma_{n}$ to ensure that the corresponding summand has a nonzero expectation. As event $A$ has the highest probability among those events that ensure this, $\mathbb{P}(A)=\mathcal{O}(1/n)$ can be used as an upper bound for the probabilities of all events described in $a$). 

The expectation of the fourth term can be upper bounded very similarly to the third one, using case separation.
For the sake of brevity, we only provide the final bound 
\vspace{-1mm}
\begin{align}
    \label{eq:eight_terms}
    &\sum_{s=1}^{\ell}\sum_{\substack{i,j,k,l\in I_s\\i < j <k<l}}\!\!\mathbb{E}\left[X_iY_{\sigma_n(i)}X_jY_{\sigma_n(j)}X_kY_{\sigma_n(k)}X_lY_{\sigma_n(l)}\right]\leq\notag\\
    &\ell\binom{\lceil \frac{n}{\ell}\rceil}{4}\kappa_0^5\Bigg[\Bigg(\frac{\ell^2\,(n-2)\,(n-3)}{n(n-1)(n-2)(n-3)}\Bigg)\notag\\
    &+\Bigg(\frac{(n-5\ell)(n-7\ell)2\ell\,4\ell}{n(n-1)(n-2)(n-3)}\Bigg)\Bigg]
    = \mathcal{O}(n^2),
\end{align}
where we can repeatedly use H\"older's inequality to upper bound the expectation terms in \eqref{eq:eight_terms} by $\kappa_0^5$.

The expectation of the fifth term of \eqref{equ:expanded_sum} can be upper bounded similarly to the second term, that is by using that at most $n^2$ expectations are summed, therefore
\vspace{-1mm}
\begin{align}
    &\sum_{s=1}^{\ell}\sum_{i,j\in I_s,i< j}\!\!\!\mathbb{E}\left[X_i^3Y_{\sigma_n(i)}^3X_jY_{\sigma_n(j)}\right]
    = \mathcal{O}(n^2).
\end{align}

\noindent Putting the five expectations together, we get}
\vspace{-1mm}
{\newtext \begin{align}
    &\sum_{s=1}^{\ell}\mathbb{E}\left(\sum_{i=1}^{n}X_iY_{\sigma_n(i)}\right)^{\!4}\!  
    =\, \mathcal{O}(n^2).
\end{align}
Consequently, we can conclude that:
\vspace{-1mm}
\begin{align}
    \frac{\ell^4}{\varepsilon^4}\sum_{n=1}^\infty \sum_{s=1}^{\ell}\mathbb{E}\Bigg\vert\frac{1}{n}\sum_{i=1}^{n}X_iY_{\sigma_n(i)}\Bigg\vert^4\! \leq\, %
    c \sum_{n=1}^\infty \frac{1}{n^2} < \infty,
\end{align}}
for some constant $c$, which completes the proof.
\end{proof}
\subsection{Strong consistency}
In the following, we state and prove our main theorem about the  strong consistency of RPS. Our proof takes several ideas from the strong consistency proof of IV-SPS \cite{volpe2015sign}. 
A major difference is that in case of sign-perturbations (such as in IV-SPS), standard SLLN type results can be applied, while for the case of RPS, we must use
Lemma \ref{lemma:permuted_sum}.
\begin{theorem}\label{thm:strong_consistency}
	Assuming \ref{assu:noise}-\ref{assu:R}, $\forall\,\varepsilon > 0,$ we have
\begin{equation}
\mathbb{P}\bigg(\bigcup_{k=1}^\infty \bigcap_{n=k}^\infty\! \big\{\,\confreg[p,n] \subseteq \mathcal{B}_{\varepsilon}(\theta^*) \big\} \bigg) =\, 1,
\end{equation}
where $\mathcal{B}_{\varepsilon}(\theta^*) \defeq \{\, \theta \in \mathbb{R}^{d}: \norm{\hspace{0.3mm}\theta - \theta^*} \le \varepsilon\, \}$.
\end{theorem}
\begin{proof}
In the first part of the proof we are going to prove that for any ``false'' parameter vector $\theta' \neq \theta^*$, we have
\begin{equation}\label{equ:lim_n_S0}
\norm{S_0(\theta')}^2\, \xrightarrow{a.s.}\, \norm{R^{-\frac{1}{2}}V(\theta^* - \theta')}^2 > 0,
\end{equation}
while, for $i \neq 0$, we have
\begin{equation}\label{equ:lim_n_Si}
\norm{S_i(\theta')}^2\, \xrightarrow{a.s.} 0,\,
\end{equation}
as $n \to \infty$. 
Recall the definitions of 
$V$ and 
$R$ from  \ref{assu:co-regressor_regressor}, \ref{assu:R}.
As a consequence of \eqref{equ:lim_n_S0} and \eqref{equ:lim_n_Si}, as $n$ grows, the rank $\mathcal{R}(\theta')$ of $\norm{S_0(\theta')}^2$ will be eventually 
$m$, therefore $\theta'$ will be (a.s.) excluded from the confidence region, as $n \to \infty$.
As a first step, we reformulate $S_0(\theta')$ and $S_i(\theta')$,
\begin{align}\label{equ:S0_expanded}
    S_0(\theta') 
    &= R_n^{-\frac{1}{2}}\frac{1}{n}\sum_{t=1}^n \psi_t\varepsilon_t(\theta')\\
    &=R_n^{-\frac{1}{2}}\frac{1}{n}\sum_{t=1}^n \psi_t\left(\varphi_t\tr  \theta^* + W_t - \varphi_t\tr \theta'\right) \notag\\
    &=R_n^{-\frac{1}{2}}\frac{1}{n}\sum_{t=1}^n \psi_t\varphi_t\tr \tilde{\theta} + \psi_tW_t\notag,
\end{align}
\vspace{-3mm}
\begin{align}\label{equ:Si_expanded}
    S_i(\theta') &= R_n^{-\frac{1}{2}}\frac{1}{n}\sum_{t=1}^{n}\psi_t\varepsilon_{\sigma_{i,n}(t)}(\theta')\\
    &= R_n^{-\frac{1}{2}}\frac{1}{n}\sum_{t=1}^n \psi_t\left(\varphi_{\sigma_{i,n}(t)}\tr  \theta^* + W_{\sigma_{i,n}(t)} - \varphi_{\sigma_{i,n}(t)}\tr \theta'\right) \notag\\
    &=R_n^{-\frac{1}{2}}\frac{1}{n}\sum_{t=1}^n \psi_t\varphi_{\sigma_{i,n}(t)}\tr \tilde{\theta} + \psi_tW_{\sigma_{i,n}(t)},\notag
\end{align}
where $\tilde{\theta} \defeq \theta^* - \theta'$. We will examine the four terms from \eqref{equ:S0_expanded} and \eqref{equ:Si_expanded} separately. 
{\newtext In case of the reference sum, we first assume that $\ell = 1$, then generalize our result to arbitrary $\ell$.}

{\em i) Reference sum first term}:
Using \ref{assu:co-regressor_regressor} and \ref{assu:R} it holds that\vspace{-0.5mm}
\begin{align}\label{equ:ref_sum_first}
    \lim_{n \to \infty}R_n^{-\frac{1}{2}}\frac{1}{n}\sum_{t=1}^n\psi_t\varphi_t\tr \tilde{\theta} =R^{-\frac{1}{2}}V\,\Tilde{\theta}\quad \text{(a.s.)}.
\end{align}

In the following, we will prove almost sure convergence to zero for every sum, therefore the term $R_n^{-\frac{1}{2}}$ can be omitted from the sums as $R_n^{-\frac{1}{2}} \xrightarrow{a.s.} R^{-\frac{1}{2}}$ (\ref{assu:R}).

{\em ii) Reference sum second term}:
Using Cantelli's SLLN element-wise it can be proven that
\vspace{-1mm}
\begin{align}
   \lim_{n \to \infty} \frac{1}{n}\sum_{t=1}^{n} \psi_tW_t = 0\quad \text{(a.s.)},
\end{align}
since $\{\psi_{t,j}W_t\}$ is an independent sequence, $\mathbb{E}[\psi_{t,j}W_t] = 0$ and $\mathbb{E}[\psi_{t,j}^4W_t^4] < \infty$ (\ref{assu:noise}, \ref{assu:co-regressor_noise}).

We conclude that, as $n \to \infty$, we have
\begin{align}
        \label{referenceS0}
        &\norm{S_0(\theta')}^2 \xrightarrow{a.s.}
        \big\|R^{-\frac{1}{2}}V\,\Tilde{\theta}\,\big\|^2 > 0.
\end{align}

{\newtext For an arbitrary $\ell$, the same construction can be used as in the proof of Lemma \ref{lemma:permuted_sum}, to decompose the sequence into 
subsequences $\{\psi_k \varphi_k\}_{k\in J_s}$
of independent variables, e.g.,
\vspace{-0.5mm}
\begin{align}
   &\frac{1}{n}\sum_{t=1}^n \psi_t\varphi_{t}\tr\tilde{\theta} = \frac{\lceil \frac{n}{\ell}\rceil}{n}\sum_{s=1}^{\ell}\bigg(\frac{1}{\lceil \frac{n}{\ell}\rceil}\sum_{t \in I^n_s}\psi_t\varphi_{t}\tr \tilde{\theta}
   \bigg),
\end{align}
where
$I^n_s$ and $J_s$ are defined above \eqref{eq:BClemma}, hence
$\lfloor n/\ell\rfloor \leq \abs{I^n_s} \leq \lceil n/\ell\rceil$.
As we 
decomposed 
$\{\psi_k \varphi_k\}$
into the sum of $\ell$ subseries, and each such subseries converges (a.s.) based on our previous arguments,
the original series converges (a.s.) to the sum of the corresponding limits. Thus, \eqref{referenceS0} is ensured for any $\ell$.
}

{\em iii) Permuted sum first term}: Notice that the summed elements in the first term of $S_i(\theta')$ \eqref{equ:Si_expanded} do not form an independent sequence, therefore well-known SLLN results cannot be applied to prove a.s. convergence. To prove that
\begin{align}
    \lim_{n \to \infty} \frac{1}{n}\sum_{t=1}^n \psi_t\varphi_{\sigma_{i,n}(t)}\tr\tilde{\theta} = 0\quad \text{(a.s.)},
\end{align}
Lemma \ref{lemma:permuted_sum} can be applied element-wise. Note that the conditions of the lemma follows from \ref{assu:regressor}, \ref{assu:co-regressor_noise} and \ref{assu:co-regressor_regressor} using that $\mathbb{E}|X|^p < \infty \Rightarrow \mathbb{E}|X|^q < \infty$ if $q \leq p$, and one can use the Cauchy–Schwarz inequality to show that the expectations of the cross-products are also bounded. Finally, the maximum of the obtained bounds can serve as $\kappa_0$.

{\em iv) Permuted sum second term}: The summed terms form an {\newtext $\ell$-independent} sequence in this case, however for every sample size $n$, a new sum is generated. Nevertheless, we can apply Lemma \ref{lemma:permuted_sum} element-wise again to prove that
\begin{align}\label{equ:perm_sum_sec}
    \lim_{n \to \infty} \frac{1}{n}\sum_{t=1}^n \psi_tW_{\sigma_{i,n}(t)}= 0\quad \text{(a.s.)},
\end{align}
since \ref{assu:noise} and \ref{assu:co-regressor_noise} satisfies the assumptions of Lemma \ref{lemma:permuted_sum}, for the same reasons as in the previous sum ({\em iii}).

From the previous derivations, we can conclude that
\begin{align}
        &\norm{S_i(\theta')}^2 \,\xrightarrow{a.s.}\,	0,
\end{align}
as $n \to \infty$, for each $i \in \{1,\dots,m-1\}$.

Now, we prove that the confidence region converges to $\theta^*$ uniformly, not just pointwise. Let us introduce 
\begin{align}
    &\Phi_n \defeq \begin{bmatrix}
        \varphi_1\tr \\
        \varphi_2\tr \\
        \vdots\\
        \varphi_n\tr 
    \end{bmatrix}\!,\qquad
        \Psi_n \defeq \begin{bmatrix}
        \psi_1\tr  \\
        \psi_2\tr \\
        \vdots\\
        \psi_n\tr 
    \end{bmatrix}\!,\qquad
    w_n \defeq \begin{bmatrix}
        W_1\\
        W_2\\
        \vdots\\
        W_n
    \end{bmatrix}\!,\notag\\
    &Q_{i,n} \defeq \frac{1}{n}\sum_{t=1}^n \psi_t\varphi_{\sigma_{i,n}(t)}\tr  = \frac{1}{n}\Psi_n\tr  P_{i,n}\Phi_n,
\end{align}
where $P_{i,n}$ is a permutation matrix corresponding to $\sigma_{i,n}$. Using the definition from \eqref{equ:def_Vn}, it holds that $V_n = \tfrac{1}{n} \Psi_n\tr  \Phi_n$.

Our previous results showed that for every $i$,
$\|S_i(\theta')\|^2$ converges (a.s.).
As a consequence,  for each realization $\omega \in \Omega$ (from an event with probability one, where $(\Omega, \mathcal{F}, \mathbb{P})$ is the underlying probability space),
and for each $\delta >0$, there exists a 
$N(\omega)>0$ such that for $n \geq N$ and $i \neq 0$,
\begin{align}
    \norm{R_n^{-\frac{1}{2}}V_n - R^{-\frac{1}{2}}V} \leq \delta,
    \hspace{2mm}\norm{\tfrac{1}{n}R_n^{-\frac{1}{2}}\Psi_n\tr w_n} \leq \delta,\\[2mm]
    \norm{R_n^{-\frac{1}{2}}Q_{i,n}} \leq \delta,
    \hspace{2mm}\norm{\tfrac{1}{n}R_n^{-\frac{1}{2}}\Psi_n\tr  P_{i,n} w_n} \leq \delta.
\end{align}
{\newtext Then, using similar reformulations as in the proof of \cite[Theorem 2]{volpe2015sign}, for all $n \geq N$, we have
\begin{align}
    &\norm{S_0(\theta')}\geq\sigma_{\text{min}}(R^{-\frac{1}{2}}V)\norm{\tilde{\theta}} -\delta\norm{\tilde{\theta}}-  \delta,
\end{align}}
\hspace{-1.5mm}where $U_{\sigma}\Sigma V_{\sigma}\tr$ is the SVD decomposition of $R^{-\frac{1}{2}}V$ and $\sigma_{\text{min}}(\cdot)$ denotes the smallest singular value. We also have
\begin{align}
    &\norm{S_i(\theta')} =  \norm{R_n^{-\frac{1}{2}}Q_{i,n}\tilde{\theta} + \tfrac{1}{n}R_n^{-\frac{1}{2}}\Psi_n\tr  P_{i,n} w_n} \leq  \notag\\
    &\norm{R_n^{-\frac{1}{2}}Q_{i,n}}\norm{\tilde{\theta}}+ \norm{\tfrac{1}{n}R_n^{-\frac{1}{2}}\Psi_n\tr  P_{i,n} w_n} \leq
    \delta\norm{\tilde{\theta}} + \delta.
\end{align}
Therefore, we have $\norm{S_i(\theta')} < \norm{S_0(\theta')}, \forall\, \theta'$ that satisfy
\begin{equation}
    \delta\norm{\tilde{\theta}} + \delta < \sigma_{\text{min}}(R^{-\frac{1}{2}}V)\norm{\tilde{\theta}} -\delta\norm{\tilde{\theta}}-  \delta,
\end{equation}
which can be reformulated as
\begin{equation}
    \mu(\delta)\, \defeq\, \frac{2\delta}{\sigma_{\text{min}}(R^{-\frac{1}{2}}V) - 2\delta} < \norm{\tilde{\theta}},
\end{equation}
therefore, those $\theta'$ for which $\mu(\delta) < \norm{\theta^* - \theta'}$ are not in the confidence region $\confreg[p,n]$, for $n \geq N$. For any 
$\varepsilon > 0$, by setting $\delta = (\varepsilon\sigma_{\text{min}}(R^{-\frac{1}{2}}V))/(2+2\varepsilon)$, we have $\confreg[p,n] \subseteq \mathcal{B}_{\varepsilon}(\theta^*)$, therefore,
the claim of the theorem follows.
\end{proof}
\section{Ellipsoidal Outer-Approximation}
The RPS-Indicator function can decide whether a given parameter is included in the confidence region. In order to give a compact representation of the whole region, we introduce a permutation-based version of the ellipsoidal outer-approximation (EOA) method \cite{Csaji2015, volpe2015sign}. The main motivation behind such constructions is that evaluating every parameter, even on a grid, to build the RPS region is computationally demanding, especially in higher dimensions. {\newtext 
The ellipsoids are constructed in a way that they have the same shape, orientation and center as the asymptotic confidence ellipsoids, see \eqref{equ:eoa_region}, only their sizes (determined by the radius parameters) are different. However, they have finite sample
guarantees. The radius computation, which is based on the ordering of the residual-permuted sums, and the construction of the ellipsoid can be derived similarly as for IV-SPS \cite{volpe2015sign}, since the sign-perturbations can be simply
replaced by permutations.}

First, we introduce a correlation approach based estimate \cite[Section 7.5]{Ljung1999}, that will be the center of the region, as\vspace{-1mm}
\begin{align}
    \hat{\theta}_n \defeq \left(\sum_{t=1}^n\psi_t\varphi_t\tr\right)^{\!-1}\sum_{t=1}^n\psi_tY_t
\end{align}
Then, the RPS outer-approximation can be given as
\begin{align}\label{equ:eoa_region}
    \outappr[n,p] \defeq \Big\{\,\theta \in \BR^d : \norm{R_n^{-\frac{1}{2}}V_n(\theta-\hat{\theta}_n)}^2 \leq r \,\Big\},
\end{align}
where $r$ is the $q$th largest solution of the following convex semi-definite programming problems, for
$i \in \{1, \dots, m-1\}$
\begin{equation}\label{equ:rps_eoa}
\begin{aligned}
    \min \quad & \gamma\\
    \textrm{s.t.} \quad & \lambda \geq 0 \\
                        & \begin{bmatrix}
                            -I + \lambda A_i & \lambda b_i\\
                            \lambda b_i\tr  & \lambda c_i + \gamma
                        \end{bmatrix} \succeq 0.
\end{aligned}
\end{equation}
In \eqref{equ:rps_eoa}, ``$\succeq 0$'' denotes positive semidefinitness and
\begin{align}
	A_i \defeq \;& I-R_n^{\frac{1}{2}\mathrm{T}}V_n^{-\mathrm{T}}Q_{i,n}\tr R_n^{-1}Q_{i,n}V_n^{-1}R_n^{\frac{1}{2}}\notag\\
	b_i \defeq \;& R_n^{\frac{1}{2}\mathrm{T}}V_n^{-\mathrm{T}}Q_{i,n}\tr R_n^{-1}(\xi_i-Q_{i,n}\hat{\theta}_n)\notag\\
	c_i \defeq \;&- \xi_i\tr R_n^{-1}\xi_i + 2\hat{\theta}_n\tr Q_{i,n}\tr R_n^{-1}\xi_i -\hat{\theta}_n\tr Q_{i,n}\tr R_n^{-1}Q_{i,n}\hat{\theta}_n\notag \\
    \xi_i \defeq & \frac{1}{n}\sum_{t=1}^n \psi_tY_{\sigma_{i,n}(t)}.
\end{align}
As $\outappr[n,p]$ is an outer approximation of $\confreg[n,p]$ it follows that
\begin{equation}
    \BP(\theta^* \in \outappr[n,p]) \,\geq\, 1-\frac{q}{m} \,=\, p,
\end{equation}
hence $\outappr[n,p]$ is a guaranteed confidence ellipsoid for any $n$.
\section{Simulation Experiments}
In this section, we illustrate the RPS approach through two numerical experiments. In the first experiment, we compared the RPS indicator and outer-approximation regions with the SPS indicator and outer-approximation regions, furthermore, with the confidence ellipsoid based on the classical asymptotic theory \cite{Ljung1999}. {\newtext We consider a $2$-dimensional FIR system
\begin{align}\label{equ:fir_exp}
    Y_t = b_1^*U_{t-1} + b_2^*U_{t-2} + W_t,
\end{align}
where $b_1^*=5$, $b_2^*=1$, $\{W_t\}$ are i.i.d. Laplacian random variables with mean $0$ and variance $1$, and the input is
$U_t = \sum_{i=1}^5c_iV_{t-i+1}$,
with $V_t \sim \mathcal{N}(0, 1)$ and $c_1=1$, $c_2=0.775$, $c_3=0.55$, $c_4=0.325$, $c_5=0.1$. 
From \eqref{equ:fir_exp}, the linear regression problem can be constructed as $\theta^* = [b_1^*,b_2^*]\tr$ and $\varphi_t = [U_{t-1},U_{t-2}]\tr$.
The sample size was $n=250$.
We chose $\psi_t = \varphi_t$ and $R_n = \tfrac{1}{n} \Phi_n\tr  \Phi_n$ to make sure that the RPS and SPS regions have the same shape and orientation as the asymptotic confidence ellipsoids.
} 
{\newtext The $0.9$-level confidence regions, with $m=10$ and $q=1$ for the RPS and SPS methods, are presented in Fig.\ref{fig:rps_sps_asym}.} It can be observed that the RPS regions are smaller than the SPS sets,
and that they are about the same size as the asymptotic confidence regions. 
This experiment indicates that RPS can outperform SPS 
sample complexity wise,
while having an advantage over the asymptotic region that it has finite sample coverage guarantees.
\begin{figure}[t]
  \vspace{1mm}
  \centering
  \includegraphics[scale=0.72]{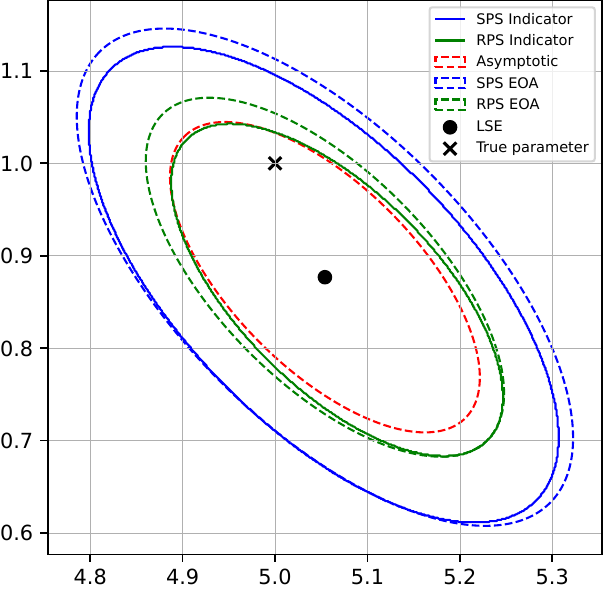}
  \caption{Comparison of 0.9-level RPS indicator, RPS EOA, SPS indicator, SPS EOA and asymptotic confidence regions for $n=250$.}
  \label{fig:rps_sps_asym}
\end{figure}

In the second experiment, we investigated the sizes of RPS regions for different sample sizes. We used the same system setting 
as in the first experiment, with the exception that $\{W_t\}$ was a sequence of i.i.d exponential random variables with parameter $0.5$, i.e.\ not a symmetric distribution about zero. Fig.\ \ref{fig:RPS_asym_diff_n} illustrates the RPS indicator and asymptotic confidence regions for $n = 200$, $n=1000$ and $n=2000$. It shows that for smaller sample sizes, RPS regions have 
smaller sizes than 
asymptotic ellipsoids, but this size difference decreases as the sample size increases. Nonetheless, RPS has exact finite sample coverage guarantees, unlike the asymptotic region.
\begin{figure}[t]
  \vspace{1mm}
  \centering
  \includegraphics[scale=0.72]{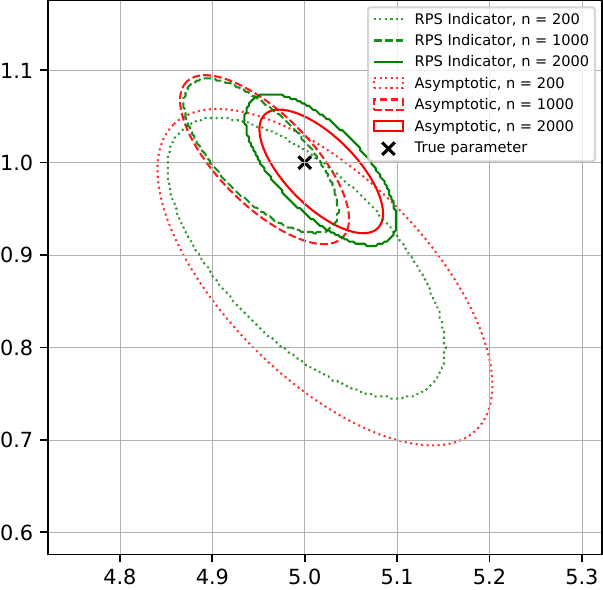}
  \caption{Comparison of 0.9-level RPS indicator and asymptotic confidence regions for $n=200$, $n=1000$ and $n=2000$.}
  \label{fig:RPS_asym_diff_n}
\end{figure}

\section{Conclusions}

In this letter, we have introduced the Residual-Permuted Sums (RPS) algorithm, motivated by a permutation-based DP method, 
as an alternative to SPS, in which the symmetricity and independence assumptions on the noises
are replaced by an i.i.d.\ condition.
RPS can construct exact, non-asymptotic, distribution-free confidence regions for the true parameters of linear regression problems.
One of the main contributions of the letter is that we proved the (uniform) strong consistency of the RPS construction under general assumptions. We also demonstrated the effectiveness of RPS empirically.

\bibliographystyle{IEEEtran}
\bibliography{sps}

\end{document}